\documentclass[letterpaper, 10 pt, conference]{ieeeconf}  
\IEEEoverridecommandlockouts    
\usepackage{comment}
\usepackage{cite}
\newcommand{\eat}[1]{}
\usepackage{booktabs}
\usepackage{color}
\usepackage [autostyle, english = american]{csquotes}
\MakeOuterQuote{"}
\usepackage{xcolor}

\usepackage{multicol}

\usepackage{mathtools}
    
\usepackage{amsmath}
\usepackage{amstext}
\usepackage{amssymb}
\usepackage{amsfonts}
\usepackage{float}

\usepackage{amsthm}  
\newtheorem{theorem}{Theorem}

\newtheorem{definition}{Definition}

\newtheorem{remark}{Remark}

\newtheorem{problem}{Problem}
\usepackage{subfig}
\usepackage{caption}
\usepackage{lipsum}

\usepackage{todonotes}

\makeatletter

\newcommand{\Rmnum}[1]{\expandafter\@slowromancap\romannumeral #1@}
\usepackage{savesym}

\usepackage{algorithm}
\usepackage{algorithmicx} 
\usepackage{algpseudocode} %
\savesymbol{AND}
\usepackage[group-separator={,},group-minimum-digits={3}]{siunitx}

\usepackage{graphicx} 
\usepackage{epsfig} 

\usepackage{times} 
\usepackage{amsmath} 
\usepackage{amssymb}  
\usepackage{comment}
\makeatletter
\let\NAT@parse\undefined
\makeatother
\usepackage{hyperref}
\hypersetup{
   colorlinks=true,
    linkcolor= blue,
    allcolors=blue,
    citecolor = blue,
    filecolor=black,      
    urlcolor=blue,
    }
\usepackage{mathrsfs}

\title{\LARGE \bf
A Barrier-Certified Optimal Coordination Framework for Connected and Automated Vehicles}
\author{Behdad Chalaki, \emph{IEEE Student Member}, Andreas A. Malikopoulos, \emph{IEEE Senior Member}%
\thanks{This research was supported by ARPAE's NEXTCAR program under the award number DE-AR0000796. This support is gratefully acknowledged.}
\thanks{The authors are with the Department of Mechanical Engineering, University of Delaware, Newark, DE 19716 USA (emails: \texttt{\{bchalaki;andreas\}@udel.edu}).}}

\begin{document}

\maketitle
\thispagestyle{empty}
\pagestyle{empty}

\begin{abstract} 
In this paper, we extend a framework that we developed earlier for coordination of connected and automated vehicles (CAVs) at a signal-free intersection by integrating a safety layer using control barrier functions. First, in our motion planning module, each CAV computes the optimal control trajectory using simple vehicle dynamics. The trajectory does not make any of the state, control, and safety constraints active. A vehicle-level tracking controller employs a combined feedforward-feedback control law to track the resulting optimal trajectory from the motion planning module. Then, a barrier-certificate module, acting as a middle layer between the vehicle-level tracking controller and physical vehicle, receives the control law from the vehicle-level tracking controller and using realistic vehicle dynamics ensures that none of the state, control, and safety constraints becomes active. The latter is achieved through a quadratic program, which can be solved efficiently in real time. We demonstrate the effectiveness of our extended framework through a numerical simulation.
\end{abstract}

\section{Introduction}
\PARstart{R}{ecent} advancement in the communication technologies and computational capabilities have been paving the way to employ fleets of connected and automated vehicles (CAVs) in transportation networks to address concerns such as safety and traffic congestion \cite{Wadud2016}. 
The influential work of Athans \cite{Athans1969} on safely coordinating CAVs at merging roadways generated significant interest in this area. Several research efforts since then have considered a two-level optimization framework.
This framework includes an \emph{upper-level} optimization that yields, for each CAV, the optimal time to exit the control zone combined with a \emph{low-level} optimization that yields for the CAV the optimal control input (acceleration/deceleration) to achieve the optimal time derived in the upper level subject to the state and control constraints. 
There have been several approaches in the literature to solve the upper-level optimization problem, including first-in-first-out (FIFO) queuing policy, heuristic  Monte Carlo tree search methods \cite{xu2019cooperative,xu2020bi}, centralized optimization techniques \cite{guney2020scheduling,hult2018optimal}, and job-shop scheduling \cite{chalaki2020TITS, fayazi2018mixed}.
Given the solution of the upper-level optimization problem, a constrained optimal control problem is solved sequentially in the low-level optimization providing the optimal control input for each CAV. 
To address the low-level optimization problem, research efforts have  used  optimal control techniques  \cite{chalaki2020TCST,mahbub2020decentralized,Malikopoulos2020,zhang2019decentralized,Kumaravel:2021wi}, which yield closed-form solutions, and model predictive control \cite{hult2018optimal,kim2014mpc,campos2014cooperative,kloock2019distributed}.  

Other approaches in the literature have explored the idea of employing control barrier functions (CBF) to ensure the satisfaction of constraints in a safety-critical system. The CBF approach handles constraints by rendering the safe sets forward invariant, which means that if the system initially starts in the safe set, it will stay in the safe set \cite{ames2016control, ames2014control}. Ames et al. \cite{ames2016control} presented a framework to unify safety constraints along with performance objectives of safety-critical system with affine control using CBFs and the control Lyapunov functions (CLFs), respectively. Under reasonable assumptions, they proved that CBF provides a necessary and sufficient condition on the forward invariance of a safe set. They demonstrated the performance of their approach on automotive applications such as adaptive cruise control and lane keeping. A comprehensive discussion of the recent effort on CBFs and their use to verify and enforce safety in the context of safety-critical controllers is provided in \cite{ames2019control}.

More recently, there have been a series of papers initially proposed by Xiao et al. \cite{xiao2021bridging,xiao2019decentralized} on using CBFs in the coordination of CAVs \cite{xiao2019decentralized,xiao2021bridging,rodriguez2022vehicle,khaled2020intersection,katriniok2021control}. Xiao et al. \cite{xiao2019decentralized}, provided a joint CBF  and CLF approach to respond to inevitable perturbation and noise in a highway merging problem. 
The authors transformed the state and control constraints of the system into the corresponding CBF constraints and solved a quadratic program (QP) at each time step. Focusing on the highway merging problem in \cite{xiao2021bridging}, the authors presented their two-step approach. First, using linearized dynamic and quadratic costs, they derived the unconstrained solution to the optimal control problem. Next, by formulating a QP at each time step, they tracked the optimal control trajectory using CLF and ensured the satisfaction of the constraints through CBF constraints. Considering an intersection scenario, Khaled et al. \cite{khaled2020intersection} applied the formulation in \cite{xiao2019decentralized} to a signal-free intersection, while Rodriguez and Fathi  \cite{rodriguez2022vehicle} employed the two-step formulation in \cite{xiao2021bridging} to an intersection with traffic lights.

In this paper, we build upon the framework introduced in \cite{Malikopoulos2020} consisting of a single optimization level aimed at both minimizing energy consumption and improving the traffic throughput. Utilizing the proposed framework, each CAV computes the optimal  unconstrained control trajectory without activating any of the state, control, and safety constraints. One direct benefit of this framework is that it avoids the inherent  implementation challenges in solving a constrained optimal control problem in real time. In this framework, we have considered that there exists a vehicle-level controller which can perfectly track the optimal unconstrained control trajectory. However, for cases where deviations  between the actual trajectory and the planned trajectory exist, some constraints of the system may become active. To address this issue, one approach is to employ a replanning mechanism \cite{chalaki2021Reseq} which introduces an indirect feedback in the system. Another approach is to consider learning these deviations and uncertainties online  \cite{chalaki2021RobustGP}.

Using CBF for safety-critical systems \cite{ames2016control}, we integrate a safety layer into our framework to guarantee that the planned trajectory does not violate any of the constraints in the system. Particularly, since safety constraints for each CAV involve the trajectory of other CAVs, inspired by the idea of environmental CBFs \cite{molnar2021safety}, we consider the evolution of other relevant CAVs in constructing our CBFs. By introducing a barrier certificate as a safety middle layer between the vehicle-level tracking controller and physical vehicle, we provide a reactive mechanism to guarantee constraint satisfaction in the system. The enhanced framework results in a QP that can be solved efficiently at each time step. This approach also allows us to consider more complex vehicle dynamics to ensure safety. 

Although several studies on coordination of CAVs  at different traffic scenarios using CBFs have been reported in the literature, the approach reported in this paper advances the state of the art in the following ways. First, in contrast to other efforts which attempt to address satisfaction of all the constraints in the system through CBFs \cite{rodriguez2022vehicle,katriniok2021control,xiao2019decentralized,khaled2020intersection}, in this paper, the motion planning module yields an optimal unconstrained trajectory which guarantees that state, control, and safety constraints are satisfied, while barrier-certificate module only intervenes if the deviations from the nominal optimal trajectory lead to violating the constraints. Second, in several research efforts using CBFs, the lateral safety is handled through imposing a FIFO queuing policy \cite{rodriguez2022vehicle,khaled2020intersection,xiao2021bridging,xiao2019decentralized}. However, in our approach, we do not consider a FIFO queuing policy. Relaxing a FIFO queuing policy is not a trivial task since it introduces a constraint with higher relative degree, which requires special analysis. 

The remainder of the paper is structured as follows. 
In Section \ref{sec:pf}, we introduce the general modeling framework. In Section \ref{sec:MP}, we present the motion planning, and in Section \ref{sec:BC}, we introduce the barrier-certificate modules. We provide simulation results in Section \ref{sec:results}, and concluding remarks in Section \ref{sec:Conclusion}.

\section{Modeling Framework} \label{sec:pf}
We consider a signal-free intersection (Fig. \ref{fig:intersection}) which includes a \textit{coordinator} that stores information about the intersection's geometry and CAVs' trajectories. The coordinator only acts as a database for the CAVs and does not make any decision. The intersection includes a \textit{{control zone}} inside of which the CAVs can communicate with the coordinator. We call the points inside the control zone where paths of CAVs intersect and a lateral collision may occur as \textit{conflict points}. Let $\mathcal{O}\subset \mathbb{N}$ be the set of conflict points, $N(t)\in\mathbb{N}$ be the total number of CAVs inside the control zone at time $t\in\mathbb{R}_{\geq0}$, and $\mathcal{N}(t)=\{1,\ldots,N(t)\}$ be the queue that designates the order in which each CAV entered the control zone.

\begin{figure}
    \centering
\includegraphics[width=0.95\linewidth]{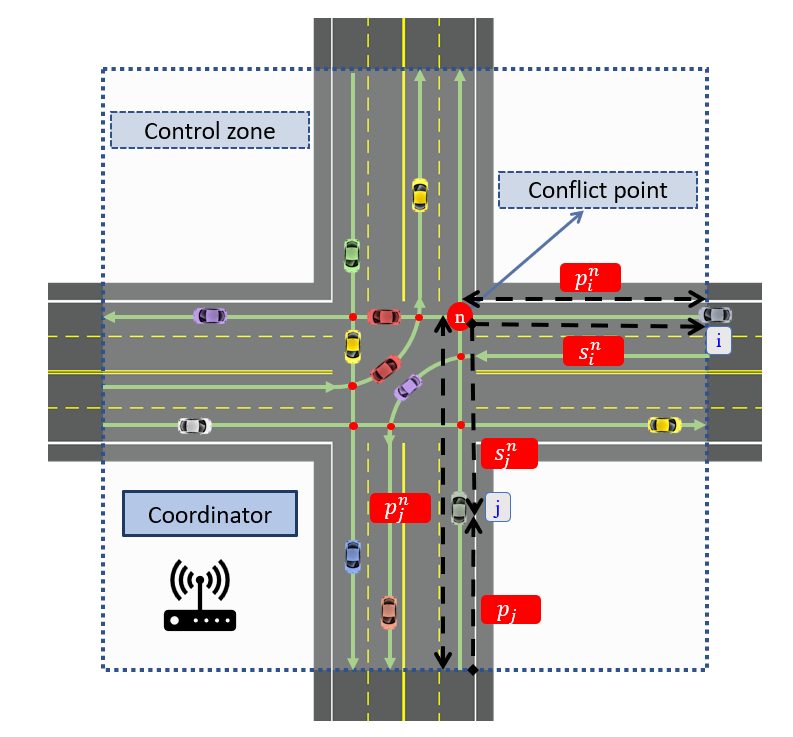}  \caption{A signal free intersection with conflict points.}
    \label{fig:intersection}
\end{figure}

Our coordination framework architecture consists of two main interconnected components called motion planning and barrier certificate (Fig. \ref{fig:flowchart}). Using the simplified dynamics of each CAV, the motion planning module which is built based on the approach reported in \cite{Malikopoulos2020} yields an optimal exit time from the control zone. The resulting optimal exit time corresponds to the unconstrained optimal control trajectory, derived using simple dynamics, and guarantees that none of the state, control, and safety constraints becomes active. The approach in \cite{Malikopoulos2020} considers that a vehicle-level tracking controller can perfectly track the resulting optimal trajectory from the motion planning module. 
In this paper, however, we no longer consider this and introduce the vehicle-level tracking controller that employs a combined feedforward-feedback control law to track the resulting optimal trajectory from the motion planning module. Then, we introduce an intermediate barrier certificate module between the vehicle-level tracking controller and physical vehicle, which takes the reference control law, and by using complex vehicle dynamics, it ensures that none of the constraints in the system are violated.
In particular, the barrier-certificate module yields a QP that can be solved at each time step onboard each CAV in real time . 

\begin{figure}
    \centering
\includegraphics[width=0.95\linewidth]{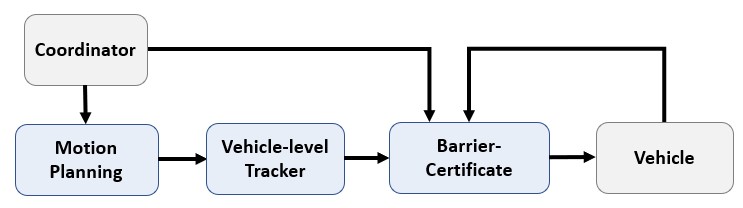}  \caption{Coordination framework architecture.}
    \label{fig:flowchart}
\end{figure}

\section{Motion planning}\label{sec:MP}
For the motion planning module, we model the dynamics of each CAV $i\in\mathcal{N}(t)$ as a double integrator
\begin{align}
\begin{aligned}\label{eq:dynamics}
\dot{p}_i(t)=v_i(t),\\
\dot{v}_i(t)=u_i(t),
\end{aligned}\end{align}
where $p_{i}(t)\in\mathcal{P}_{i}$, $v_{i}(t)\in\mathcal{V}_{i}$, and
$u_{i}(t)\in\mathcal{U}_{i}$ denote position, speed, and control input at $t$, respectively. The sets $\mathcal{P}_{i}$,
$\mathcal{V}_{i}$, and $\mathcal{U}_{i}$, for $i\in\mathcal{N}(t),$
are compact subsets of $\mathbb{R}$. Let $t_{i}^{0}\in\mathbb{R}_{\geq 0}$ be the time that CAV $i\in\mathcal{N}(t)$
enters the control zone, and $t_{i}^{f}>t_i^0\in\mathbb{R}_{\geq 0}$ be the time that CAV $i$ exits the control zone. For each CAV $i\in\mathcal{N}(t)$, the control input and speed are bounded by 
\begin{align}
    u_{i,\min}&\leq u_i(t)\leq u_{i,\max}, \label{eq:uconstraint} \\
    0< v_{\min}&\leq v_i(t)\leq v_{\max} \label{eq:vconstraint},
\end{align}
where $u_{i,\min},u_{i,\max}$ are the minimum and maximum control inputs and $v_{\min},v_{\max}$ are the minimum and maximum speed limits, respectively.

To guarantee rear-end safety between CAV $i\in\mathcal{N}(t)$ and a preceding CAV $k\in\mathcal{N}(t)$, we impose the following constraint,
\begin{gather}\label{eq:rearend}
 p_k(t)-p_i(t)\geq  \underbrace{\gamma + \varphi~ v_i(t)}_{\delta_i(t)},
\end{gather}
where $\delta_i(t)$ is the safe speed-dependent distance, while $\gamma$ and $\varphi\in\mathbb{R}_{>0}$ are the standstill distance and reaction time, respectively. 
\begin{definition}\label{def:distance2Node}
For CAV $i\in\mathcal{N}(t)$, and a conflict point $n\in\mathcal{O}$, $s_i^n:\mathbb{R}_{\geq0}\rightarrow\mathbb{R}$ is the function that gives the distance between CAV $i$ and conflict point $n$ (Fig. \ref{fig:intersection}), and it is given by
\begin{align}
    s_i^n(t) = p_i^n-p_i(t),\quad \forall t\in[t_i^0, t_i^f],
\end{align}
where $p_{i}^n$ is the distance of the conflict point $n\in\mathcal{O}$ from the point that  CAV $i$ enters the control zone.
\end{definition}

Let CAV $j\in\mathcal{N}(t)$ be a CAV that has already planned its trajectory which might cause a lateral collision with CAV $i$. 
CAV $i$ can reach at conflict point $n$ either after or before CAV $j$. 
In the first case, we have
\begin{equation} \label{eq:lateralAfter}
    s_i^n(t)+s_j^n(t) \geq \delta_i(t), \quad \forall t\in[t_i^0, t_j^n],
\end{equation}
where $t_j^n$ is the known time that CAV $j$ reaches the conflict point $n$, i.e., position $p_j^n$. The intuition in \eqref{eq:lateralAfter} is that at $t_j^n$, $s_j^n$ is equal to zero based on Definition \ref{def:distance2Node}, and CAV $i$ should maintain at least a safe distance $\delta_i(t)$ from the conflict point $n$. However, for $t\in[t_i^0, t_j^n)$, $s_j^n$ is a positive number, and hence $s^n_i(t)$ needs to be greater than ${\delta_i(t)-s_j^n(t)}$.
Similarly, in the second case, where CAV $i$ reaches the conflict point $n$ before CAV $j$, we have 
\begin{equation} \label{eq:lateralBefore}
    s_i^n(t)+s_j^n(t) \geq \delta_j(t), \quad \forall t\in[t_i^0, t_i^n],
\end{equation}
where $t_i^n$ is determined by the trajectory planned by CAV $i$. 

Since $ 0 < v_{\min} \leq v_i(t)$, the position $p_i(t)$ is a strictly increasing function.
Thus, the inverse $t_i\left(\cdot\right) = p_i^{-1}\left(\cdot\right)$ exists and it is called the \textit{time trajectory} of CAV $i$ \cite{Malikopoulos2020}. Hence, we have $t_i^n = p_i^{-1}\left(p_i^n\right)$. Therefore, for each candidate path of CAV $i$, there exists a unique time trajectory which can be evaluated at conflict point $n$ to find the time $t_i^n$ that CAV $i$ reaches at conflict point $n$. 

By moving all terms in \eqref{eq:lateralAfter} to the LHS, we obtain ${s_i^n(t)+s_j^n(t)-\delta_i(t)\geq 0}$. Constraint \eqref{eq:lateralAfter} is satisfied, if $\min(s_i^n(t)+s_j^n(t)-\delta_i(t))\geq 0$ in the interval $[t_i^0, t_j^n]$. Likewise, constraint \eqref{eq:lateralBefore} is satisfied if $\min(s_i^n(t)+s_j^n(t)-\delta_j(t))\geq 0$ in the interval $[t_i^0, t_i^n]$.
However, to ensure the lateral safety between CAV $i$ and CAV $j$ at conflict point $n$, either \eqref{eq:lateralAfter} or \eqref{eq:lateralBefore} must be satisfied, and thus we impose the following lateral safety constraint on CAV $i$ 
\begin{align}
    \max \Bigg\{ &\min_{t\in[t_i^0, t_j^n]} \{s_i^n(t)+s_j^n(t)-\delta_i(t)\}, \notag\\
            &\min_{t\in[t_i^0, t_i^n]} \{ s_i^n(t)+s_j^n(t)-\delta_j(t) \}   \Bigg\} \geq 0. \label{eq:lateralMinSafety}
\end{align}

Next, we briefly review the motion planning module that includes the single-level optimization framework for coordination of CAV reported in \cite{Malikopoulos2020}. In this framework, each CAV $i$ communicates with the coordinator to solve a time minimization problem, which determines $t_i^f$, i.e., the time that CAV $i$ must exit the control zone. The optimal exit time $t_i^f$ corresponds to the unconstrained optimal control trajectory which guarantees that none of the state, control, and safety constraints becomes active. This trajectory is communicated back to the coordinator, so that the subsequent CAVs receive this information and plan their trajectories accordingly. Using the unconstrained optimal control trajectory in $[t_i^0, t_i^f]$ which does not activate any of the state, control, and safety constraints, we essentially avoid the inherent implementation challenges in solving a constrained optimal control in real time which requires piecing constrained and unconstrained arcs together \cite{mahbub2020Automatica-2,chalaki2020experimental}.

To formally define the motion planning problem, we first start with the unconstrained optimal control solution of CAV $i$, which has the following form \cite{Malikopoulos2020}
\begin{align} \label{eq:optimalTrajectory}
    u_i(t) &= 6 a_i t + 2 b_i, \notag \\
    v_i(t) &= 3 a_i t^2 + 2 b_i t + c_i, \\
    p_i(t) &= a_i t^3 + b_i t^2 + c_i t + d_i, \notag
\end{align}
where $a_i, b_i, c_i$, and $d_i$ are constants of integration. CAV $i$ must also satisfy the boundary conditions
\begin{align}
     p_i(t_i^0) &= 0,\quad  v_i(t_i^0)= v_i^0 , \label{eq:bci}\\
     p_i(t_i^f)&=p_i^f,\quad u_i(t_i^f)=0, \label{eq:bcf}
\end{align}
where $u_i(t_i^f)=0$ because the speed at the exit of the control zone is not specified \cite{bryson1975applied}.
The details of the derivation of the unconstrained solution are discussed in \cite{Malikopoulos2020}.  

Next, we formally define the motion planning problem to minimize the exit time from the control zone. 

\begin{problem}\label{prb:mintfProblem}
Each CAV $i\in\mathcal{N}(t)$ solves the following optimization problem at $t_i^0$, upon entering the control zone 
\begin{align}\label{eq:tif}
    &\min_{t_i^f\in \mathcal{T}_i(t_i^0)} t_i^f \\
    &\emph{subject to: }\eqref{eq:rearend},\eqref{eq:lateralMinSafety}, \eqref{eq:optimalTrajectory}, \eqref{eq:bci},\eqref{eq:bcf}, \notag
\end{align}
where the compact set $\mathcal{T}_i(t_i^0)$ is the set of feasible solution of CAV $i\in\mathcal{N}(t)$ for the exit time computed at $t_i^0$ using the speed and control input constraints \eqref{eq:uconstraint}-\eqref{eq:vconstraint}, initial condition \eqref{eq:bci}, and final condition \eqref{eq:bcf}. The derivation of this compact set is discussed in \cite{chalaki2020experimental}.
\end{problem}
Solving Problem \ref{prb:mintfProblem}, CAV $i$ derives the optimal exit time, $t_i^f$, corresponding to an optimal trajectory, $\Bar{u}_i(t),\Bar{v}_i(t)$ and $\Bar{p}_i(t)$, which satisfies all the state, control, and safety constraints.

\begin{equation}
    u_i^{ref}(t) = \Bar{u}_i(t) + k_p \cdot( \Bar{p}_i(t) - p_i(t)) + k_v\cdot(\Bar{p}_i(t) - p_i(t)),\label{eq:feedbak-feedforward}
\end{equation}
where $p_i(t)$ and $v_i(t)$ are current observed position and speed of CAV $i$; respectively, while $k_p,k_v\in \mathbb{R}_{>0}$ are feedback control gains. 

\begin{remark}
In this paper, we consider that CAVs solve Problem \ref{prb:mintfProblem} upon entering the control zone. However, one can consider the case in which CAVs re-solve their motion planning problem either periodically or based on an occurrence of a certain event such as the entrance of a new CAV in the control zone as described in \cite{chalaki2021Reseq}.
\end{remark}

\section{Barrier-certificate}\label{sec:BC}
In this section, we present our barrier-certificate module which is a middle layer between the vehicle-level tracking controller and physical vehicle. 
In this module, we consider more realistic model to describe the dynamics of each CAV $i\in\mathcal{N}(t)$ as follows
\begin{align}
\begin{aligned}\label{eq:dynamicsCBF}
\dot{p}_i(t) &= v_i(t),\\
\dot{v}_i(t) &= u_i(t)-\frac{F_r(v_i(t))}{m_i}.
\end{aligned}\end{align}
Let $F_r\in\mathbb{R}_{\geq 0}$ correspond to all resisting forces including longitudinal aerodynamic drag force and rolling resistance force at tires, while $m_i\in\mathbb{R}_{\geq 0}$ is the mass of CAV \cite{khalil2002nonlinear,rajamani2011vehicle}. The net resisting force typically is approximated as a quadratic function of the CAV's speed \cite[Chapter~2]{khalil2002nonlinear}, i.e.,
\begin{gather}
    F_r(v_i(t)) = \beta_0 + \beta_1~v_i(t) + \beta_2~v^2_i(t),
\end{gather}
where $\beta_0,\beta_1,\beta_2\in\mathbb{R}_{\geq 0}$ are all constant parameters that can be computed empirically. We write \eqref{eq:dynamicsCBF} in a control-affine, vector form as 
\begin{align}\label{eq:controlAffine}
 \dot{\mathbf{x}}_{i}(t) &= 
 \underbrace{\begin{bmatrix}
v_i(t) \\
-\frac{F_r(v_i(t))}{m_i} 
\end{bmatrix}}_{ \mathbf{f}_i(\mathbf{x}_{i}(t))}
+ 
\underbrace{\begin{bmatrix}
0 \\
1 
\end{bmatrix}}_{\mathbf{g}_i(\mathbf{x}_{i}(t))} u_i(t),
\end{align}
where $\mathbf{x}_{i}(t)=\left[p_{i}(t), v_{i}(t)\right]^\top \in \mathcal{P}_i\times\mathcal{V}_i$ denotes the state of the CAV $i$ at $t$. Note that $\mathbf{f}_i$ and $\mathbf{g}_i$ are globally Lipschitz functions, which results in global existence and uniqueness of the solution of \eqref{eq:controlAffine} if $u_i$ is also globally Lipschitz \cite[Chapter~3]{khalil2002nonlinear}.

\subsection{Preliminary Materials}
In this section, we review some basic definitions and results from \cite{ames2016control,ames2019control} adapted appropriately to reflect our notation. Inspired by the idea of environmental CBFs \cite{molnar2021safety}, we construct a CBF for the cases in which the constraint of the CAV is coupled to the dynamics of other CAVs, such as lateral safety and rear-end safety constraints. To simplify notation, we discard the argument of time in our state and control variables whenever it does not create confusion.

Next, we define the safe set of a constraint that depends only on the state of a CAV.
 
\begin{definition} \label{def:SafeSet} For CAV $i\in\mathcal{N}(t)$, the safe set $\mathcal{C}$ is a zero-superlevel set of a continuously differentiable function $h:\mathcal{P}_i\times\mathcal{V}_i\rightarrow\mathbb{R}$,
\begin{align}\label{eq:safeSet}
    \mathcal{C} &=\{\mathbf{x}_i\in \mathcal{P}_i\times\mathcal{V}_i: h(\mathbf{x}_i)\geq 0\}.
\end{align}
\end{definition}
For those cases where a constraint of CAV $i$ depends also on another CAV $j$, i.e., in rear-end safety and lateral safety constraints, we define the coupled safe set next.

\begin{definition}\label{def:CoupledSafeSet} For CAV $i\in\mathcal{N}(t)$, the coupled safe set $\mathcal{C}^{\prime}$  with CAV $j\in\mathcal{N}(t)$ is a zero-superlevel set of a continuously differentiable function $z:D\subseteq	(\mathcal{P}_i\times\mathcal{V}_i)\times(\mathcal{P}_j\times\mathcal{V}_j)\rightarrow\mathbb{R}$,
\begin{align}\label{eq:safeSet}
    \mathcal{C}^{\prime} &=\{(\mathbf{x}_i,\mathbf{x}_j)\in D: z(\mathbf{x}_i,\mathbf{x}_j)\geq 0\}.
\end{align}
\end{definition}

Next, we define the safety of the CAV $i$, with longitudinal dynamics \eqref{eq:controlAffine}, with respect to the safe set $\mathcal{C}$.

\begin{definition} CAV $i$ with the longitudinal dynamics given by \eqref{eq:controlAffine} is safe with respect to the safe set $\mathcal{C}$ if the set $\mathcal{C}$ is forward-invariant, namely, if $\mathbf{x}_i(t_i^0)\in\mathcal{C}$, $\mathbf{x}_i(t)\in\mathcal{C}$ for all  $t\geq t_i^0$.
\end{definition}
Similarly, we define safety with respect to the coupled safe set $\mathcal{C}^{\prime}$. 

\begin{definition} CAV $i$ with longitudinal dynamics given by \eqref{eq:controlAffine} is safe with respect to the coupled safe set $\mathcal{C}^{\prime}$ with CAV $j$, if the set $\mathcal{C}^{\prime}$ is forward-invariant, namely, if $(\mathbf{x}_i(t_i^0),\mathbf{x}_j(t_i^0)) \in\mathcal{C}^{\prime}$, $(\mathbf{x}_i(t),\mathbf{x}_j(t))\in\mathcal{C}^{\prime}$ for all $t\geq t_i^0$.
\end{definition}
Next, we need to define the extended class $\mathcal{K}_{\infty}$ function.
\begin{definition}
A strictly increasing function $\alpha:\mathbb{R}\rightarrow\mathbb{R}$ with $\alpha(0)=0$, is an extended class $\mathcal{K}_{\infty}$ function.
\end{definition}

\begin{definition}[\!\cite{ames2016control}]\label{def:CBFnotCoupled}
Let $\mathcal{C}$ be a safe set for CAV $i\in\mathcal{N}(t)$ for a continuously differentiable function $h:\mathcal{P}_i\times\mathcal{V}_i\rightarrow\mathbb{R}$. The function $h$ is a CBF if there exists an extended class $\mathcal{K}_{\infty}$ function $\alpha(\cdot)$ such that for all $\mathbf{x}_i\in\mathcal{C}$
\begin{align}
    \sup_{u_i\in\mathcal{U}_i} \dot{h}(\mathbf{x}_i,u_i) \geq -\alpha(h(\mathbf{x}_i)),
\end{align}
where 
\begin{align} \label{eq:gradh}
 \dot{h}(\mathbf{x}_i,u_i) &= \nabla h(\mathbf{x}_i) \cdot \dot{\mathbf{x}}_i,
\end{align}
and $\dot{\mathbf{x}}_i$ is given by \eqref{eq:controlAffine}.
\end{definition}
\begin{remark}
We can also write \eqref{eq:gradh} in terms of Lie derivatives as follows
\begin{align}
    \dot{h}(\mathbf{x}_i,u_i) = L_{\mathbf{f}_i} h(\mathbf{x}_i) +  L_{\mathbf{g}_i} h(\mathbf{x}_i) u_i,
\end{align}
where 
\begin{align}
    L_{\mathbf{f}_i} h(\mathbf{x}_i) &= \left[\frac{\partial h(\mathbf{x}_i)}{\partial p_i }, \frac{\partial h(\mathbf{x}_i)}{\partial v_i }\right]^\top\cdot\mathbf{f}_i(\mathbf{x}_i),\\
    L_{\mathbf{g}_i} h(\mathbf{x}_i) &= \left[\frac{\partial h(\mathbf{x}_i)}{\partial p_i }, \frac{\partial h(\mathbf{x}_i)}{\partial v_i }\right]^\top\cdot\mathbf{g}_i(\mathbf{x}_i).
\end{align}
\end{remark}

\begin{theorem}[\! \cite{ames2016control}]\label{th:cbfTheorem}
Let $\mathcal{C}$ be a safe set for CAV $i\in\mathcal{N}(t)$ for a continuously differentiable function $h:\mathcal{P}_i\times\mathcal{V}_i\rightarrow\mathbb{R}$. If $h$ is a CBF on $\mathcal{P}_i\times\mathcal{V}_i$, then any Lipschitz continuous controller $u_i:\mathcal{P}_i\times\mathcal{V}_i \rightarrow \mathcal{U}_i$ such that $u_i(\mathbf{x}_i)\in \mathcal{A}_{h}(\mathbf{x}_i)$  renders the safe set $\mathcal{C}$ forward invariant, where 
\begin{align}\label{eq:cbfSafe}
    \mathcal{A}_{h}(\mathbf{x}_i) = \{u_i \in \mathcal{U}_i:\nabla h(\mathbf{x}_i) \cdot \dot{\mathbf{x}}_i \geq -\alpha(h(\mathbf{x}_i))\}.
\end{align}
\end{theorem}

Inspired by the idea of environmental CBF \cite{molnar2021safety}, which considers the evolution of environment state in analyzing safety, we consider the evolution of other relevant CAVs in constructing the CBF for CAV $i$.

\begin{definition}\label{def:CBFCoupled}
Let $\mathcal{C}^{\prime}$ be a coupled safe set for CAV $i$ and $j\in\mathcal{N}(t)$ for a continuously differentiable function $z:D\subseteq	(\mathcal{P}_i\times\mathcal{V}_i)\times(\mathcal{P}_j\times\mathcal{V}_j)\rightarrow\mathbb{R}$. The function $z$ is a CBF if there exists an extended class $\mathcal{K}_{\infty}$ function $\alpha(\cdot)$ such that for all $(\mathbf{x}_i,\mathbf{x}_j)\in\mathcal{C}^{\prime}$
\begin{align}
    \sup_{u_i\in\mathcal{U}_i} \dot{z}(\mathbf{x}_i,u_i, \mathbf{x}_j, u_j) \geq -\alpha(z(\mathbf{x}_i,\mathbf{x}_j)),
\end{align}
where 
\begin{align} \label{eq:gradz}
 \dot{z}(\mathbf{x}_i,u_i, \mathbf{x}_j, u_j) &= \nabla_{\mathbf{x}_i} z(\mathbf{x}_i,\mathbf{x}_j) \cdot
 \underbrace{(\mathbf{f}_i(\mathbf{x}_{i}) + \mathbf{g}_i(\mathbf{x}_{i}) u_i)}_{\dot{\mathbf{x}}_i}\nonumber\\ 
 &+\nabla_{\mathbf{x}_j} z(\mathbf{x}_i, \mathbf{x}_j) \cdot  \underbrace{(\mathbf{f}_j(\mathbf{x}_{j}) + \mathbf{g}_j(\mathbf{x}_{j}) u_j)}_{\dot{\mathbf{x}}_j},\\
 \nabla_{\mathbf{x}_i} z(\mathbf{x}_i, \mathbf{x}_j)&=\left[\frac{\partial z(\mathbf{x}_i,\mathbf{x}_j)}{\partial p_i }, \frac{\partial z(\mathbf{x}_i,\mathbf{x}_j)}{\partial v_i }\right]^\top, \\
 \nabla_{\mathbf{x}_j} z(\mathbf{x}_i, \mathbf{x}_j)&=\left[\frac{\partial z(\mathbf{x}_i,\mathbf{x}_j)}{\partial p_j }, \frac{\partial z(\mathbf{x}_i,\mathbf{x}_j)}{\partial v_j }\right]^\top. 
\end{align}
\end{definition}
\begin{remark}
Note that in our decentralized coordination framework, CAV $i\in\mathcal{N}(t)$ plans its trajectory after CAV $j\in\mathcal{N}(t)$, which means that $\mathbf{x}_j$ and $u_j$ are available to CAV $i$ through the coordinator. 
\end{remark}

\begin{theorem}\label{th:CBFCouples}
Let $\mathcal{C}^{\prime}$ be a coupled safe set for CAV $i\in\mathcal{N}(t)$ and $j\in\mathcal{N}(t)$ for a continuously differentiable function $z:D\subseteq	(\mathcal{P}_i\times\mathcal{V}_i)\times(\mathcal{P}_j\times\mathcal{V}_j)\rightarrow\mathbb{R}$. If $z$ is a CBF on $D$, then any Lipschitz continuous controller $u_i:D \rightarrow \mathcal{U}_i$ such that $u_i(\mathbf{x}_i,\mathbf{x}_j)\in \mathcal{A}_{z}(\mathbf{x}_i,\mathbf{x}_j)$  renders the coupled safe set $\mathcal{C}^{\prime}$ forward invariant, where 
\begin{align}\label{eq:cbfCoupledSafe}
    \mathcal{A}_{z}(\mathbf{x}_i,\mathbf{x}_j) &= \{u_i \in \mathcal{U}_i:\nabla_{\mathbf{x}_i} z(\mathbf{x}_i,\mathbf{x}_j) \cdot \dot{\mathbf{x}}_i \nonumber\\ &+\nabla_{\mathbf{x}_j} z(\mathbf{x}_i,\mathbf{x}_j) \cdot \dot{\mathbf{x}}_j\geq -\alpha(z(\mathbf{x}_i,\mathbf{x}_j))\}.
\end{align}
\end{theorem}
\begin{proof}
The proof is  similar to the one in \cite[Theorem 2]{molnar2021safety}.
By considering the new state $\mathbf{X}_{i,j}$ as stacked state of $\mathbf{x}_i$ and $\mathbf{x}_j$, and applying Theorem \ref{th:cbfTheorem}, the result follows. 
\end{proof}

Using CBFs, we can map all of the constraints from the states for CAV $i\in\mathcal{N}(t)$ to the control input as, formally derived next.

\subsection{Constructing CBFs}
In this section, we construct CBFs for \eqref{eq:vconstraint}-\eqref{eq:lateralBefore}.
\subsubsection{Speed limits}
For the speed constraint \eqref{eq:vconstraint} of CAV $i$, we consider 
\begin{align}
    h_1(\mathbf{x}_i) &= v_{\max} - v_i,\\ 
    h_2(\mathbf{x}_i) &= v_i - v_{\min}.
\end{align}
From Definition \ref{def:CBFnotCoupled} and choosing $\alpha_q(x) = \lambda_q x,~\lambda_q\in\mathbb{R}_{>0}, q\in\{1,2\}$, we have $h_1(\mathbf{x}_i)$ and $h_2(\mathbf{x}_i)$ as CBFs to ensure satisfying the speed limit constraint. Then, from Theorem \ref{th:cbfTheorem}, any control input $u_i$ should satisfy the following
\begin{align}
    u_i\leq \frac{F_r(v_i)}{m_i} + \lambda_1(v_{\max}-v_i) ,\label{eq:CBFVMax}\\ 
    u_i\geq \frac{F_r(v_i)}{m_i} - \lambda_2(v_i-v_{\min})\label{eq:CBFVMin}.
\end{align}
\subsubsection{Rear-end safety}Since rear-end safety constraint depends on both states of CAV $i$ and $k\in\mathcal{N}(t)$, we have 
\begin{align}
    z_1(\mathbf{x}_i,\mathbf{x}_k) &= p_k -p_i -(\gamma + \varphi\cdot v_i).
\end{align}
From Definition \ref{def:CBFCoupled} and choosing $\alpha_3(x) = \lambda_3 x,~\lambda_3\in\mathbb{R}_{>0}$, we have $z_1(\mathbf{x}_i,\mathbf{x}_k)$ as a CBF to guarantee satisfying the rear-end safety constraint. Next, we use the result of Theorem \ref{th:cbfTheorem} to derive the condition on control input that needs to be satisfied.  The gradient of $z_1$ is
\begin{align}
     \nabla_{\mathbf{x}_i} z_1(\mathbf{x}_i, \mathbf{x}_k)&=\left[-1, -\varphi\right]^\top,\label{eq:delXi}\\
    \nabla_{\mathbf{x}_k}z_1(\mathbf{x}_i, \mathbf{x}_k)&=\left[1,0\right]^\top\label{eq:delXk}.
\end{align}
Taking the dot product of \eqref{eq:delXi} and \eqref{eq:delXk} with $\dot{\mathbf{x}}_i$ and $\dot{\mathbf{x}}_k$; respectively, yields
\begin{align}
     \nabla_{\mathbf{x}_i} z_1(\mathbf{x}_i, \mathbf{x}_k) \cdot \dot{\mathbf{x}}_i &=-v_i-\varphi(-\frac{F_r(v_i)}{m_i}+u_i),\\
     \nabla_{\mathbf{x}_k} z_1(\mathbf{x}_i, \mathbf{x}_k) \cdot \dot{\mathbf{x}}_k &=v_k.
\end{align}
Using the result of Theorem \ref{th:CBFCouples}, the control input $u_i$ should satisfy the following condition in order to satisfy the rear-end safety constraint,
\begin{align}\label{eq:cbfRear-end}
    u_i \leq \frac{1}{\varphi} \left[\lambda_3 (p_k-p_i-(\gamma+\varphi v_i)) + v_k-v_i\right]+\frac{F_r(v_i)}{m_i}.
\end{align}
\subsubsection{Lateral safety} For the lateral-safety constraint \eqref{eq:lateralAfter}, when $t_i^n>t_j^n$, we have
\begin{align}
    z_2(\mathbf{x}_i,\mathbf{x}_j) &= s_i^n+s_j^n -\delta_i \nonumber\\
    &=(p_i^n - p_i) + (p_j^n-p_j) - (\gamma + \varphi\cdot v_i).
\end{align}
By choosing $\alpha_4(x) = \lambda_4 x,~\lambda_4\in\mathbb{R}_{>0}$,  $z_2(\mathbf{x}_i,\mathbf{x}_j)$ is a CBF to guarantee satisfying the lateral safety constraint, which implies
\begin{align}
     \nabla_{\mathbf{x}_i} z_2(\mathbf{x}_i, \mathbf{x}_j)&=\left[-1, -\varphi\right]^\top,\\
    \nabla_{\mathbf{x}_j}z_2(\mathbf{x}_i, \mathbf{x}_j)&=\left[-1,0\right]^\top.
\end{align}
Taking the dot product of above equations with $\dot{\mathbf{x}}_i$ and $\dot{\mathbf{x}}_j$; respectively, yields  
\begin{align}
     \nabla_{\mathbf{x}_i} z_2(\mathbf{x}_i, \mathbf{x}_j) \cdot \dot{\mathbf{x}}_i &=-v_i-\varphi(-\frac{F_r(v_i)}{m_i}+u_i),\\
     \nabla_{\mathbf{x}_j} z_2(\mathbf{x}_i, \mathbf{x}_j) \cdot \dot{\mathbf{x}}_j &=-v_j.
\end{align}
For this case, the control input $u_i$ should satisfy the following condition in order to satisfy constraint \eqref{eq:lateralAfter},
\begin{align}\label{eq:cbfLateralAfter}
    u_i \leq \frac{1}{\varphi} \left[\lambda_4(s_i^n+s_j^n -\delta_i ) -(v_i+v_j)\right]+\frac{F_r(v_i)}{m_i}.
\end{align}
For the lateral-safety constraint \eqref{eq:lateralBefore}, we have
\begin{align}\label{eq:z3}
    z_3(\mathbf{x}_i,\mathbf{x}_j) &= s_i^n+s_j^n -\delta_j \nonumber\\
    &=(p_i^n - p_i) + (p_j^n-p_j) - (\gamma + \varphi\cdot v_j).
\end{align}
However, since $\dot{z}_3$ does not depend on $u_i$, \eqref{eq:z3} cannot be a valid CBF for CAV $i$. These type of constraints are called constraints with higher relative degree $r>1$. For example, the relative degree of \eqref{eq:lateralBefore} is equal to $2$.
A complete analysis of handling higher relative degree constraints in general cases is given in \cite{xiao2019control}. 

Next, we use a higher order CBF based on \cite[Definition~7, Theorem~5]{xiao2019control}, and extend it to our case with coupled constraints. We first form a series of functions $\psi_q:D\subseteq	(\mathcal{P}_i\times\mathcal{V}_i)\times(\mathcal{P}_j\times\mathcal{V}_j)\rightarrow\mathbb{R}$, $q=\{0,1,2\}$ as 
\begin{align}\label{eq:psis}
    \psi_0(\mathbf{x}_i,\mathbf{x}_j) &= z_3(\mathbf{x}_i,\mathbf{x}_j), \notag\\
    \psi_1(\mathbf{x}_i,\mathbf{x}_j)  & =\dot{\psi}_0(\mathbf{x}_i,\mathbf{x}_j)+ \alpha_5(\psi_0(\mathbf{x}_i,\mathbf{x}_j)),\\
    \psi_2(\mathbf{x}_i,\mathbf{x}_j) & =\dot{\psi}_1(\mathbf{x}_i,\mathbf{x}_j)+\alpha_6(\psi_1(\mathbf{x}_i,\mathbf{x}_j)),\notag
\end{align}
where $\alpha_5(\cdot)$ and $\alpha_6(\cdot)$ are extended class $\mathcal{K}_{\infty}$ functions. The zero-superlevel sets of $\psi_0$ and $\psi_1$ are given by  
\begin{align}
\mathcal{C}^{\prime}_1 &=\{(\mathbf{x}_i,\mathbf{x}_j)\in D: \psi_0(\mathbf{x}_i,\mathbf{x}_j)\geq 0\},\\
\mathcal{C}^{\prime}_2 &=\{(\mathbf{x}_i,\mathbf{x}_j)\in D: \psi_1(\mathbf{x}_i,\mathbf{x}_j)\geq 0\}.
\end{align}
Based on \cite[Definition~7]{xiao2019control}, if there exist extended class $\mathcal{K}_{\infty}$ functions $\alpha_5(\cdot)$ and $\alpha_6(\cdot)$ such that ${ \psi_2(\mathbf{x}_i,\mathbf{x}_j)\geq 0}$ for all $(\mathbf{x}_i,\mathbf{x}_j) \in \mathcal{C}^{\prime}_1 \cap \mathcal{C}^{\prime}_2$, $z_3(\mathbf{x}_i,\mathbf{x}_j)$ is a higher order CBF.
From \cite[Theorem~5]{xiao2019control}, if $(\mathbf{x}_i(t_i^0),\mathbf{x}_j(t_i^0)) \in \mathcal{C}^{\prime}_1 \cap \mathcal{C}^{\prime}_2$, then any Lipschitz continuous controller $u_i:D\rightarrow\mathbb{R}$ such that $u_i(\mathbf{x}_i,\mathbf{x}_j)\in \mathcal{A}_{\psi}(\mathbf{x}_i,\mathbf{x}_j)$ renders the set $\mathcal{C}^{\prime}_1 \cap \mathcal{C}^{\prime}_2$ forward invariant, where 
\begin{align}\label{eq:cbfCoupledSafeRelativeDegree2}
    \mathcal{A}_{\psi}(\mathbf{x}_i,\mathbf{x}_j) &= \{u_i \in \mathcal{U}_i: \psi_2(\mathbf{x}_i,\mathbf{x}_j)\geq 0\}.
\end{align}

\begin{theorem}
The allowable set of control actions that renders the set $\mathcal{C}^{\prime}_1 \cap \mathcal{C}^{\prime}_2$ forward invariant, $\mathcal{A}_{\psi}$, is given by
\begin{align}\label{eq:cbflateralBefore}
u_i&\leq -\lambda_5(v_i+v_j)+\frac{F_r(v_i)}{m_i}+\frac{F_r(v_j)}{m_j}-\frac{\varphi \beta_1 F_r(v_j)}{m_j^2}\nonumber\\
&-\frac{2\varphi \beta_2 v_j F_r(v_j)}{m_j^2}+(\frac{\varphi \beta_1+2\varphi \beta_2 v_j}{m_j}-\lambda_5\varphi-1)u_j\nonumber\\
&-\varphi \dot{u}_j + \lambda_6~\psi_1.
\end{align}
\end{theorem}
\begin{proof}
By choosing $\alpha_q(x) = \lambda_q x,~\lambda_q\in\mathbb{R}_{>0}, q\in\{5,6\}$, we have 
\begin{align}\label{eq:psi_1}
    &\psi_1(\mathbf{x}_i,\mathbf{x}_j)=\nabla_{\mathbf{x}_i}z_3(\mathbf{x}_i,\mathbf{x}_j)\cdot\dot{\mathbf{x}}_i + \nabla_{\mathbf{x}_j}z_3(\mathbf{x}_i,\mathbf{x}_j)\cdot\dot{\mathbf{x}}_j\notag\\
    &+\lambda_5 z_3(\mathbf{x}_i,\mathbf{x}_j),
\end{align}
where 
\begin{align}
    \nabla_{\mathbf{x}_i}z_3(\mathbf{x}_i,\mathbf{x}_j)\cdot\dot{\mathbf{x}}_i &= -v_i,\label{eq:delz3xi}\\
    \nabla_{\mathbf{x}_j}z_3(\mathbf{x}_i,\mathbf{x}_j)\cdot\dot{\mathbf{x}}_j &= -v_j+\varphi\frac{F_r(v_j)}{m_j}-\varphi u_j.\label{eq:delz3xj}
\end{align}
Substituting \eqref{eq:z3},\eqref{eq:delz3xi}, and \eqref{eq:delz3xj} in \eqref{eq:psi_1} yields
\begin{align}\label{eq:psi_1Plugged}
    &\psi_1(\mathbf{x}_i,\mathbf{x}_j,u_j, F_r(v_j)) = -v_i-v_j+\varphi\frac{F_r(v_j)}{m_j}-\varphi u_j\notag\\
    &+\lambda_5[p_i^n - p_i + p_j^n-p_j - \gamma - \varphi\cdot v_j].
\end{align}

Next, we derive the full time derivative of \eqref{eq:psi_1Plugged} in order to construct $\psi_2(\mathbf{x}_i,\mathbf{x}_j)$ in \eqref{eq:psis},
\begin{align}\label{eq:psi_1dot}
    &\dot{\psi}_1(\mathbf{x}_i,\mathbf{x}_j) = \nabla_{\mathbf{x}_i}\psi_1\cdot\dot{\mathbf{x}}_i + \nabla_{\mathbf{x}_j}\psi_1\cdot\dot{\mathbf{x}}_j+\frac{\partial \psi_1}{\partial u_j}~\dot{u}_j  \notag\\
    &+\frac{\partial \psi_1}{\partial F_r(v_j)}\frac{\partial F_r(v_j)}{\partial v_j}\dot{v}_j,
\end{align}
where 
\begin{align}
    \nabla_{\mathbf{x}_i}\psi_1\cdot\dot{\mathbf{x}}_i &= -\lambda_5 v_i + \frac{F_r(v_i)}{m_i}-u_i, \label{eq:delPsi_1}\\
    \nabla_{\mathbf{x}_j}\psi_1\cdot\dot{\mathbf{x}}_j &=-\lambda_5 v_j+\frac{F_r(v_j)}{m_j} + \lambda_5\varphi\frac{F_r(v_j)}{m_j}\notag\\&-u_j-\lambda_5u_j,\\
    \frac{\partial \psi_1}{\partial u_j}~\dot{u}_j &=-\varphi~\dot{u}_j,\\
    \frac{\partial \psi_1}{\partial F_r(v_j)}\frac{\partial F_r(v_j)}{\partial v_j}\dot{v}_j &= \frac{\varphi}{m_j}(\beta_1+2\beta_2v_j)\cdot(-\frac{F_r(v_j)}{m_j}+u_j).\label{eq:delPsi_2}
\end{align}
By substituting \eqref{eq:psi_1dot}-\eqref{eq:delPsi_2} into \eqref{eq:psis}, we derive $\psi_2(\mathbf{x}_i,\mathbf{x}_j)$ from \eqref{eq:psis}, which can then be used to construct the condition for the control input $u_i$ based on \eqref{eq:cbfCoupledSafeRelativeDegree2}, and the proof is complete.  
\end{proof}

As described in Section \ref{sec:MP}, to guarantee the lateral safety between CAV $i\in\mathcal{N}(t)$ and CAV $j\in\mathcal{N}(t)$ at conflict point $n\in\mathcal{O}$, either \eqref{eq:lateralAfter} or \eqref{eq:lateralBefore} must be satisfied. Thus, depending on the the arrival time at conflict point $n$ for CAV $i$ and $j$ ($t_i^n$ and $t_j^n$, respectively), we must satisfy \eqref{eq:cbfLateralAfter} or \eqref{eq:cbflateralBefore} as follows
\begin{equation}\label{eq:CBFLateral}
      \left\{ \begin{array}{ll}
u_i\leq A, \quad \text{if}~~t_i^n>t_j^n\\
u_i\leq B, \quad \text{if}~~t_i^n<t_j^n
           \end{array}, \right. 
\end{equation}
where 
\begin{align}
    A &= \frac{1}{\varphi} \left[\lambda_4(s_i^n+s_j^n -\delta_i ) -(v_i+v_j)\right]+\frac{F_r(v_i)}{m_i},\\
    B &= \lambda_5(v_i+v_j)+\frac{F_r(v_i)}{m_i}+\frac{F_r(v_j)}{m_j}-\frac{\varphi \beta_1 F_r(v_j)}{m_j^2}\nonumber\\
&-\frac{2\varphi \beta_2 v_j F_r(v_j)}{m_j^2}+(\frac{\varphi \beta_1+2\varphi \beta_2 v_j}{m_j}-\lambda_5\varphi-1)u_j\nonumber\\
&-\varphi \dot{u}_j + \lambda_6~\psi_1.
\end{align}

 Next, we formulate an optimization problem based on QP for our barrier-certificate module. This QP can be solved at discrete time step to verify the reference control input $u^{ref}_i(t)$, resulting from the vehicle-level tracking controller. In case of a potential violation, QP minimally modifies the control input to guarantee the satisfaction of all  constraints.    

\begin{problem}\label{prb:QPProblem}
Each CAV $i\in\mathcal{N}(t)$ at time $t$ observes its state $\mathbf{x}_i$ and accesses the states and control inputs, $\mathbf{x}_j$ and $u_j$, respectively, of neighbour CAVs. Then, $i$ solves the following optimization problem to find the safe control input.
\begin{align}\label{eq:QP}
     u^*_i(t)=&~\underset{u_i(t)}{\arg\min}~ \frac{1}{2}\lVert u_i(t) - u^{ref}_i(t) \rVert^2 \\
    \emph{subject to: }& \notag\\
    & \eqref{eq:uconstraint},\eqref{eq:CBFVMax},\eqref{eq:CBFVMin},\eqref{eq:cbfRear-end},\eqref{eq:CBFLateral},\notag
\end{align}
where each pertaining constraint \eqref{eq:vconstraint}-\eqref{eq:lateralBefore} for CAV $i$ are mapped to the control input constraint using the appropriate CBFs \eqref{eq:cbfSafe}, \eqref{eq:cbfCoupledSafe}, or \eqref{eq:cbfCoupledSafeRelativeDegree2}. Note that $u^{ref}_i(t)$ is the combined feedforward-feedback control law to track the resulting optimal trajectory from the motion planning module \ref{sec:MP}.
\end{problem}

Since the control input is bounded, the feasibility of the QP in Problem \ref{prb:QPProblem} can be ensured by choosing appropriate $\lambda_q\in\mathbb{R}_{\geq 0}$ for class $\mathcal{K}_{\infty}$ functions $\alpha_q(x)=\lambda_q x$, $q\in \mathbb{N}$. Note that in this paper, we chose linear class $\mathcal{K}_{\infty}$ functions; however, one may decide to choose a different form for their class $\mathcal{K}_{\infty}$. Analyzing and studying the effects of the choice of $\mathcal{K}_{\infty}$ on the control input's feasible space is left for future work.

\section{Simulation Results}\label{sec:results}
To show the performance of our barrier-certified coordination framework, we investigate the coordination of $24$ CAVs at a signal-free intersection shown in Fig.\ref{fig:intersection}. The CAVs enter the control zone from $6$ different paths (Fig. \ref{fig:intersection}) with a total rate of $3600$ veh/hour while their initial speed is uniformly distributed between $12$ m/s and $14$ m/s. We consider the length of the control zone and road width to be $212$ m and $3$ m, respectively.  
The rest of the parameters for the simulation are   $v_{\min}=0.2$ m/s, $v_{\max}=20$ m/s, $u_{\max}=2$ m/s$^2$, $u_{\min}=-2$ m/s$^2$, $\gamma=2.5$ m, $\varphi=0.5$ s $k_p=k_v=1.5$, $\Delta t=0.1$ s. 
We used \texttt{lsqlin} in Matlab to solve Problem \ref{prb:QPProblem} and \texttt{ODE45} to integrate the vehicle dynamics. Videos from our simulation can be found at the supplemental site, \url{https://sites.google.com/view/ud-ids-lab/BCOCF}.

Figs. \ref{fig:controlInput}-\ref{fig:Speed} demonstrate the control input, position, and speed for a selected CAV in the simulation. The blue line in Fig. \ref{fig:controlInput} shows the reference control input from the feedforward-feedback control law \eqref{eq:feedbak-feedforward}, and the dashed red line denotes the resulting optimal control trajectory from the motion planning module. The black line shows the applied control input at each time step resulting from the Solution of Problem \ref{prb:QPProblem}.
It can be seen that around $16.5$ s the barrier-certificate module overrides the reference control input in order to satisfy the speed limit constraint. 
The actual trajectory of the vehicle in Figs. \ref{fig:Position} and \ref{fig:Speed} is computed by integrating the realistic vehicle dynamics \eqref{eq:controlAffine} and applying the solution of Problem \ref{prb:QPProblem} at each time step. Our proposed framework tracks the resulting optimal trajectory from the motion planning module, while it ensures that none of the state, control, and safety constraints becomes active.

\begin{figure}[ht]
    \centering
\includegraphics[width=0.95\linewidth]{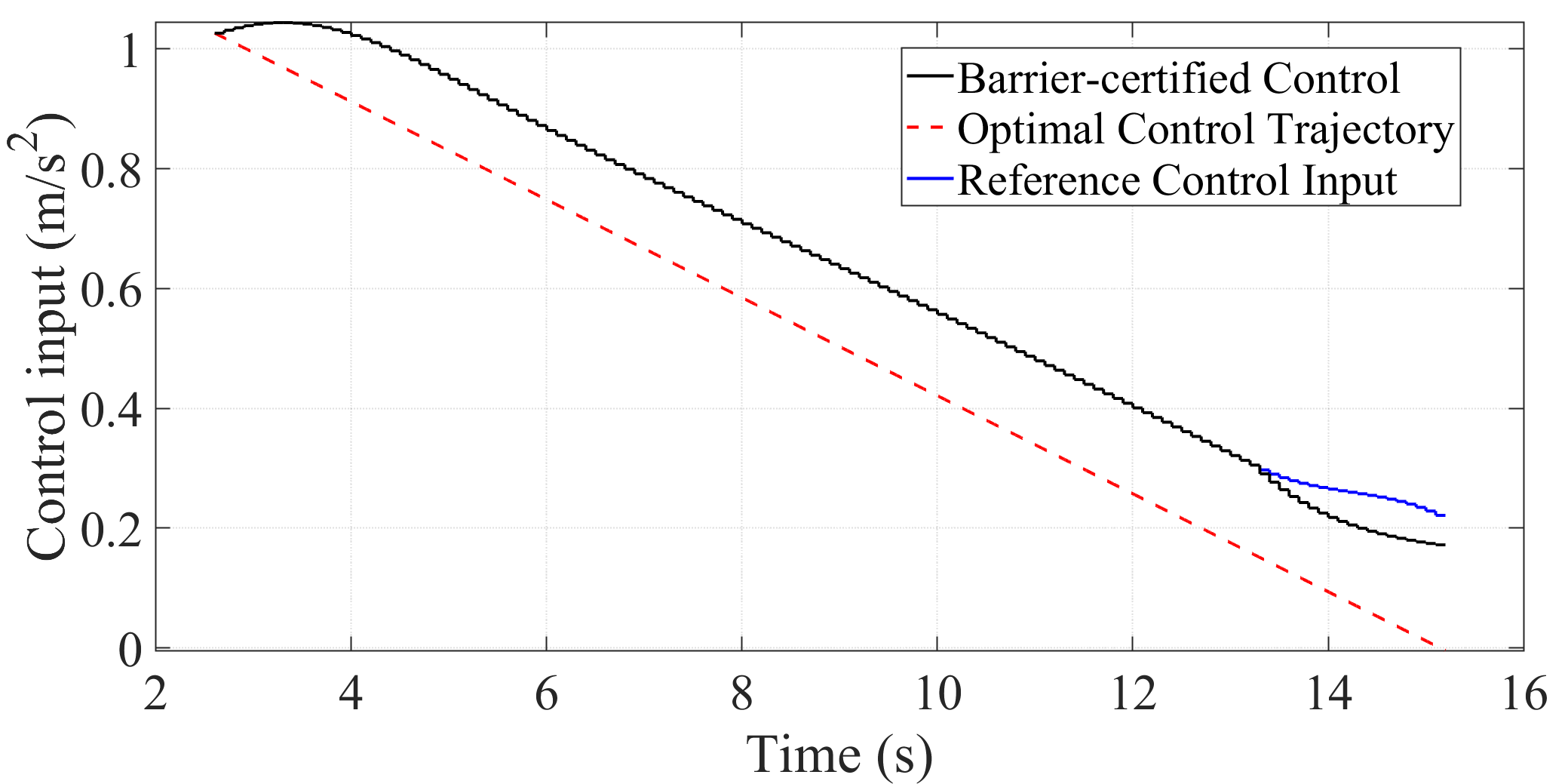}    \caption{Control input for a selected CAV.}
    \label{fig:controlInput}
\end{figure}

\begin{figure}[ht]
    \centering
\includegraphics[width=0.95\linewidth]{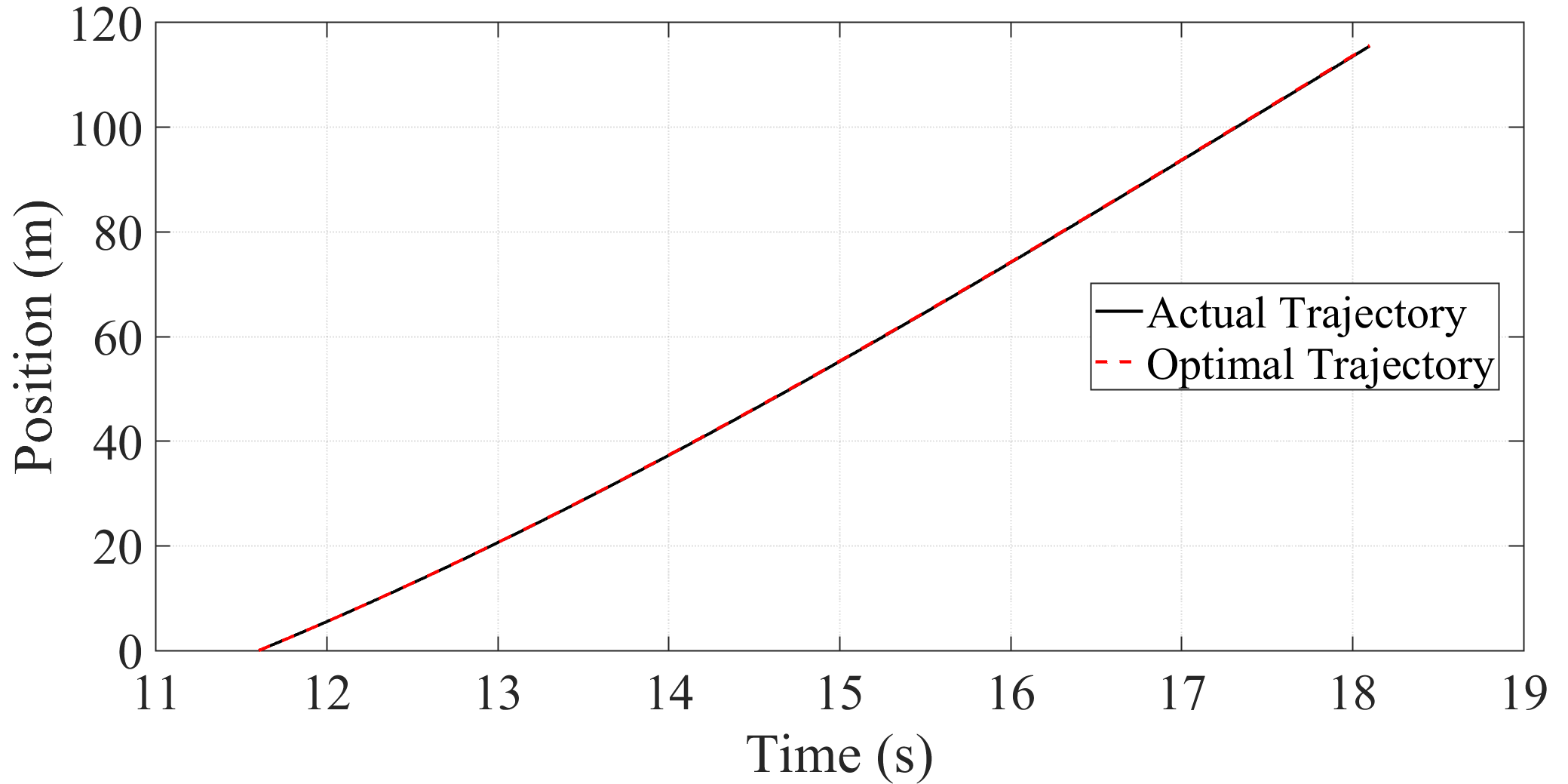}    \caption{Actual and optimal position trajectory for a selected CAV.}
    \label{fig:Position}
\end{figure}

\begin{figure}[ht]
    \centering
\includegraphics[width=0.95\linewidth]{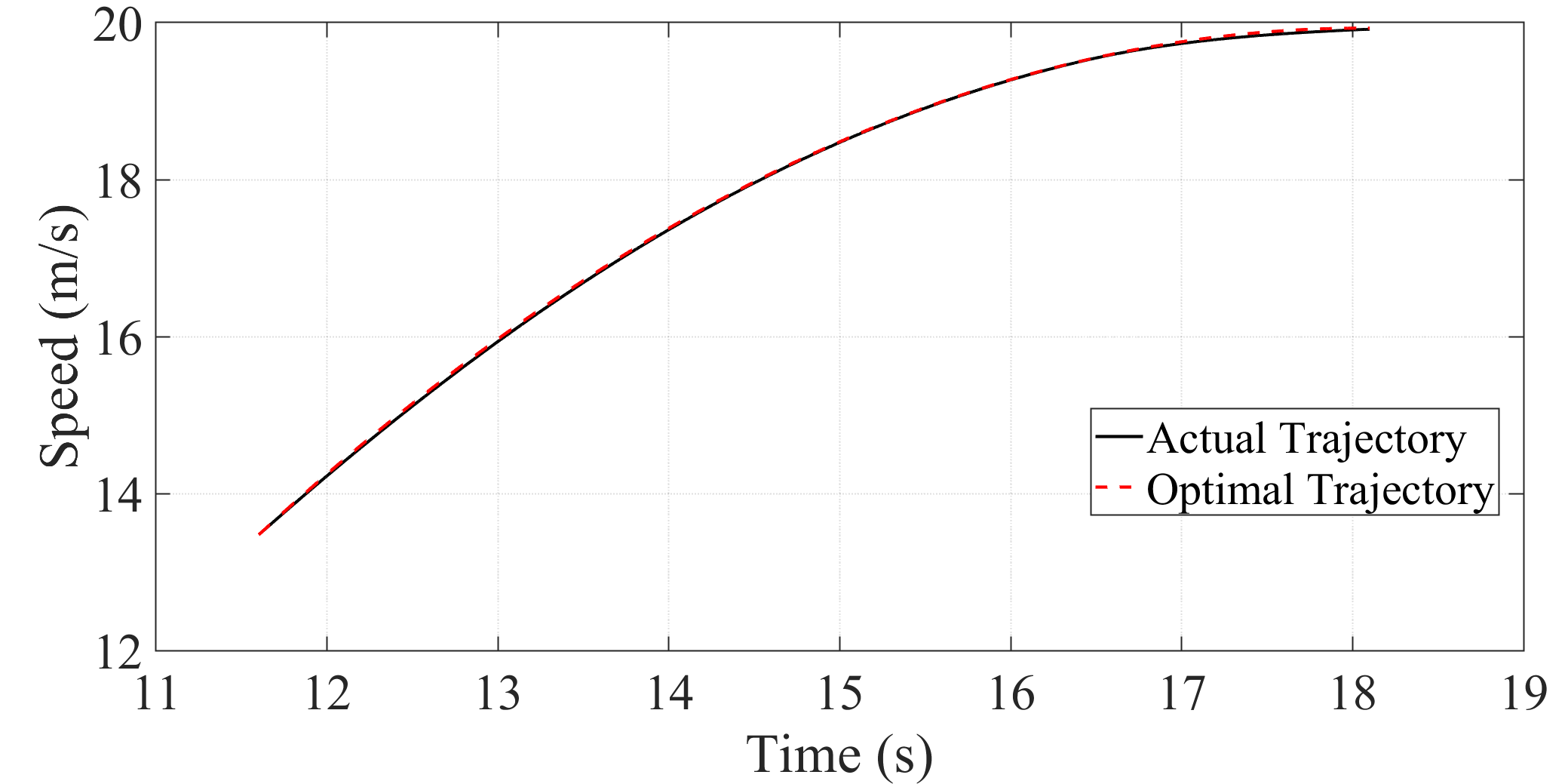}    \caption{Actual and optimal speed trajectory for a selected CAV.}
    \label{fig:Speed}
\end{figure}
The mean and standard deviation of computation times of the motion planning  and barrier-certificate modules in our proposed framework are listed in Table \ref{tbl:computation}. It shows that our framework is computationally feasible.
\begin{table}[ht]
\caption{The mean and standard deviation of computation times for each module.}
\vspace{0.5em}
\centering
\begin{tabular}{c|c|c} \label{tbl:computation}
      &  Mean (s)& Standard deviation (s)
    \\
    \toprule
    Motion planning & $0.029$& $0.0331$ \\
    Barrier-certificate & $0.0063 $&$0.0026$
\end{tabular}
\end{table}

\section{Concluding Remarks and Discussion} \label{sec:Conclusion}

In this paper, we enhanced the motion planning framework for coordination of CAVs introduced in \cite{Malikopoulos2020} through employing CBFs to provide an additional safety layer and ensure satisfaction of all  constraints in the system. By using the proposed framework in the motion planning module, each CAV first uses simple longitudinal dynamics to derive the optimal control trajectory without activating any constraint.
In a real physical system, we require a vehicle-level controller to track the resulting optimal trajectory. However, due to the inherent deviations between the actual trajectory and the planned trajectory, the system's constraints may become active. We addressed this issue by introducing a barrier-certificate module based on a more realistic dynamics as a safety middle layer between the vehicle-level tracking controller and physical vehicle to provide a reactive mechanism to guarantee constraint satisfaction in the system. Future work should validate this framework beyond simulation using a physical system.  

\bibliographystyle{IEEEtran.bst} 
\bibliography{reference/IDS_Publications_03262022.bib, reference/ref.bib,reference/CBF_REF.bib}

\end{document}